\documentclass[conference]{IEEEtran}
\usepackage{graphicx}
\usepackage{epstopdf}
\usepackage[latin1]{inputenc}
\usepackage{amsmath,amssymb,amscd,latexsym,dsfont}
\usepackage[english]{babel}
\usepackage{tabularx,cite}
\usepackage[figurename=Fig.]{caption}
\usepackage{float}
\usepackage{tensor}
\usepackage{multicol}
\usepackage{caption}
\usepackage{psfrag}
\usepackage{epstopdf}
\usepackage{comment,psfrag,subfigure,enumerate}
\usepackage{comment,psfrag,subfigure,enumerate}
\usepackage{graphicx}
\usepackage{color}
\usepackage{amsmath,bm}
\usepackage{amsthm}
\usepackage{bigints}
\usepackage{lipsum}
\usepackage{relsize}
\usepackage{calligra}
\usepackage{mathtools}

\newtheorem{mylem}{Lemma}

\ifCLASSINFOpdf
\else
\fi
\begin{document}
%
% paper title
% Titles are generally capitalized except for words such as a, an, and, as,
% at, but, by, for, in, nor, of, on, or, the, to and up, which are usually
% not capitalized unless they are the first or last word of the title.
% Linebreaks \\ can be used within to get better formatting as desired.
% Do not put math or special symbols in the title.
\title{Wireless-powered relaying with\\ finite block-length codes}

% author names and affiliations
% use a multiple column layout for up to three different
% affiliations

\author{\IEEEauthorblockN{$\text{Mahdi Haghifam}^\star$, $\text{Behrooz Makki}^\dagger$, $\text{Masoumeh Nasiri-Kenari}^\star$, $\text{Tommy Svensson}^\dagger$, $\text{Michele Zorzi}^\ddagger $ }

\IEEEauthorblockA{
$^\star$Electrical Engineering Department, Sharif University of Technology,
Tehran, Iran. \\ haghifam$\_$mahdi@ee.sharif.edu, mnasiri@sharif.edu \\
$^\dagger$Department of Signals and Systems, Chalmers University of Technology, Gothenburg, Sweden. \\ $\left\lbrace \text{behrooz.makki, tommy.svensson} \right\rbrace\text{@chalmers.se}$\\
$^\ddagger$ Department of Information Engineering, University of
Padova, Padova, Italy.\\ $\text{zorzi@dei.unipd.it}$
}
}

% conference papers do not typically use \thanks and this command
% is locked out in conference mode. If really needed, such as for
% the acknowledgment of grants, issue a \IEEEoverridecommandlockouts
% after \documentclass haghifam_mahdi@ee.sharif.edu mnasiri@sharif.edu

% for over three affiliations, or if they all won't fit within the width
% of the page, use this alternative format:
% 
%\author{\IEEEauthorblockN{Michael Shell\IEEEauthorrefmark{1},
%Homer Simpson\IEEEauthorrefmark{2},
%James Kirk\IEEEauthorrefmark{3}, 
%Montgomery Scott\IEEEauthorrefmark{3} and
%Eldon Tyrell\IEEEauthorrefmark{4}}
%\IEEEauthorblockA{\IEEEauthorrefmark{1}School of Electrical and Computer Engineering\\
%Georgia Institute of Technology,
%Atlanta, Georgia 30332--0250\\ Email: see http://www.michaelshell.org/contact.html}
%\IEEEauthorblockA{\IEEEauthorrefmark{2}Twentieth Century Fox, Springfield, USA\\
%Email: homer@thesimpsons.com}
%\IEEEauthorblockA{\IEEEauthorrefmark{3}Starfleet Academy, San Francisco, California 96678-2391\\
%Telephone: (800) 555--1212, Fax: (888) 555--1212}
%\IEEEauthorblockA{\IEEEauthorrefmark{4}Tyrell Inc., 123 Replicant Street, Los Angeles, California 90210--4321}}

% use for special paper notices
%\IEEEspecialpapernotice{(Invited Paper)}
% make the title area
\maketitle

% As a general rule, do not put math, special symbols or citations
% in the abstract
\begin{abstract}
This paper studies the outage probability and the throughput of amplify-and-forward relay networks with wireless information and
energy transfer. We use some recent results on finite block-length codes to analyze the system performance in the cases with short codewords. Specifically, the time switching relaying and the power splitting relaying protocols are considered for energy and information transfer. We derive tight approximations for the outage probability/throughput. Then, we analyze the outage probability in asymptotically high signal-to-noise ratios. Finally, we use numerical results
to confirm the accuracy of our analysis and to evaluate the system performance in different scenarios. Our results indicate that, in delay-constrained scenarios, the codeword length affects the outage probability/throughput of the joint energy and information transfer systems considerably.     
\end{abstract}

% no keywords

% For peer review papers, you can put extra information on the cover
% page as needed:
% \ifCLASSOPTIONpeerreview
% \begin{center} \bfseries EDICS Category: 3-BBND \end{center}
% \fi
%
% For peerreview papers, this IEEEtran command inserts a page break and
% creates the second title. It will be ignored for other modes.
\IEEEpeerreviewmaketitle

\vspace{-3mm}
\section{Introduction}
\vspace{-1mm}
Relay-assisted communication is one of the promising techniques
that have been proposed for wireless networks. The main idea of a relay network is to improve the data transmission efficiency by implementation of intermediate relay nodes which support data
transmission from a source to a destination. These devices are usually powered by fixed but limited batteries. Thus, wireless networks may suffer from the short lifetime
and require periodic battery replacement/recharging. However,
the battery replacement may be costly or infeasible in, e.g., toxic environments. It has been recently proposed to use
radio-frequency (RF) signals as a means of wireless energy
transfer. Significant advances in circuit design for RF energy transfer make the usage of energy transfer a viable solution for prolonging the lifespan of wireless networks, e.g., \cite{tutorwet,cirs}.\par
Amplify-and-forward (AF) and decode-and-forward (DF) cooperative networks with wireless energy and information transfer are investigated in \cite{nasir13,jad2,poor,poor_rand,finite1}. In \cite{nasir13}, two relaying protocols, namely, power splitting relaying (PSR) and time switching relaying (TSR), with simultaneous wireless information and energy transfer, are proposed and evaluated in terms of system throughput. Moreover, \cite{poor} derives different power allocation strategies for energy harvesting DF relay networks with multiple source-destination pairs and a single energy
harvesting relay. Multi-relay networks with information and energy transfer are also studied in \cite{poor_rand,finite1} using stochastic geometry.\par
In the literature, so far, the main focus has been on investigating the energy harvesting wireless networks performance based on Shannon's results on the achievable rates. To be specific, most results are obtained under the assumption that the achievable rate is given by $\log(1+x)$ with $x$ standing for the signal-to-interference-plus-noise ratio (SINR) at the receiver. This is an acceptable assumption when long codewords are used. On the other hand, in many applications, such as Internet of Things in which employing energy harvesting nodes is deemed as a core element, the codewords are required to be short to meet latency requirements \cite{delay1}. As a result, analyzing the system performance using Shannon's capacity formula does not provide realistic results in the aforementioned scenarios. Recently, an accurate approximation of the achievable rate with finite blocklength was presented in \cite{bl_1}. Using the results in \cite{bl_1}, in \cite{mak2} and \cite{behmain} the performance of incremental redundancy HARQ and spectrum sharing cognitive radio were analyzed, respectively. Also, \cite{fb_rele} studies the capacity of DF relay networks with fixed energy supply. Thus, it is interesting to study the performance of relay networks with wireless energy and information transfer in the presence of finite-length codewords.\par
In this paper, we investigate the outage probability and the throughput of AF relay networks with wireless energy and information transfer. Two
well-known protocols for energy and information transfer, namely, TSR and PSR are considered. We use the recent results of \cite{bl_1} on the achievable rates of finite block-length codes to analyze the system performance. We derive two tight closed-form approximations for the outage probability of the network (Lemmas 2-3). Moreover, we find the outage probability in the high signal-to-noise ratio (SNR) regime and the effect of the block length on the outage probability (Lemma 4). It is worth noting that the results in this paper can be readily extended to analyze AF relay networks with fixed energy supply, as well.  \par
The numerical and analytical results show that, for a given number of information nats, there exists an optimal finite codeword length that maximizes the throughput (Fig. \ref{fig:thr_L}). Furthermore, at high SNRs, the outage probability decreases linearly with $\frac{\log(\text{SNR})}{\text{SNR}}$ (Lemma 4). Finally, for a broad range of parameter settings, a lower outage probability is achieved by TSR compared to PSR (Fig. \ref{fig:err_P}). On the other hand, for a fixed number of information nats, the PSR protocol outperforms the TSR approach in terms of throughput (Figs. \ref{fig:thr_L},\ref{fig:thr_ceof}). 
\section{SYSTEM MODEL} \label{sec:sysmod}
We consider a relay-assisted communication setup consisting of a source, a relay and a destination. The source and the destination nodes have constant energy supply. On the other hand, the relay has no fixed energy supply and receives its energy from the source RF signal wirelessly. The relay is assumed to employ the AF protocol, because of its hardware simplicity\cite{rel_snr}. The channel coefficients of the source-relay and the relay-destination links are denoted by $h_{\text{sr}}$ and $h_{\text{rd}}$, respectively. The channel coefficients remain constant during the channel coherence time and then change according to their probability density functions (PDFs). Also, we define the channel gains as $g_{\text{sr}}=|h_{\text{sr}}|^2$ and $g_{\text{rd}}=|h_{\text{rd}}|^2$. The results are obtained for Rayleigh fading channels where the channel gains PDFs are given by $f_{\vartheta}\left(x\right)=e^{-x}$ for $\vartheta=\{\text{sr},\text{rd}\}$\footnote{ To simplify the presentation of the analytical results, the results are
presented for normalized fading random variables where $\mathbb{E}\left[g_{\text{sr}}\right]=\mathbb{E}\left[g_{\text{rd}}\right]=1$
with $\mathbb{E}\left[\cdot\right]$ being the expectation operator. Extension of the results to the cases with different average channel gains is straightforward.}. Slotted operation is assumed, where each time slot has duration $T$ seconds. Also, the occupied bandwidth is $W$ Hertz. Hence, in each time slot, each codeword length can maximally be $L \triangleq TW$ channel uses.\par
The system performance is analyzed for the PSR and TSR protocols. These schemes have been first introduced in \cite{nasir13} and we adapt them based on our problem formulation/channel model as follows.  
\subsection{Time Switching Relaying }\label{subsec:TSR}
In this protocol, the energy transfer and data communication protocol is in three phases as follows.
In the first phase, of length $\theta L$ channel uses (cu), $\theta\in[0,1],$  the relay harvests energy from the source transmitted energy signal. Let $P_\text{s}$ denote the source transmission power. Then, the baseband signal model in this period is given by
\begin{equation} \label{eq:bas_chan_en}
\bm{y}^{\text{e}}_{\text{r}}=\sqrt{P_\text{s}}h_{\text{sr}}\bm{x}_{\text{s}}^{\text{e}} + \bm{z}^{\text{e}}_{\text{r}}.
\end{equation}
Here, $\bm{y}^{\text{e}}_{\text{r}},\bm{z}^{\text{e}}_{\text{r}},\bm{x}_{\text{s}}^{\text{e}}\in \mathbb{C}^{L\theta}$. Also, $\bm{x}_{\text{s}}^{\text{e}}$ is the unit-power source energy signal and $\bm{z}^{\text{e}}_{\text{r}}$ is the additive white Gaussian noise (AWGN) of the energy receiver of the relay and follows $\mathcal{CN}\left(0,\left(\sigma^{\text{e}}_\text{r}\right)^2\bm{I}_{L\theta}\right)$, where $\bm{I}_n$ denotes the identity matrix of size $n$. For simplicity, we set $\sigma^{\text{e}}_\text{r}=0$. This is motivated by the fact that in many practical energy harvesting systems the harvested energy due to noise is negligible, e.g., \cite{cirs}. In this way, the energy harvested by the relay at the end of Phase 1 is given by $E_\text{r}^{\text{TSR}}=L\theta\eta g_\text{sr}P_{\text{s}}$, with $0 < \eta \leq 1 $ representing the efficiency factor of the energy harvesting circuit. With no loss of generality, we set $\eta=1$. If the relay is successfully powered up in the first phase, the second phase of length $\frac{1-\theta}{2}L$ cu starts and the source sends information to the relay. Let $\bm{x}_{\text{s}}^{\text{i}}\in \mathbb{C}^{\frac{L\left(1-\theta\right)}{2}}$ and $\bm{y}_\text{r}^\text{i} \in \mathbb{C}^{\frac{L\left(1-\theta\right)}{2}}$ denote the source information signal and its corresponding signal received by the relay, respectively. Hence, the channel is modeled as
\begin{equation}\label{eq:bas_chan_it1}
\bm{y}^{\text{i}}_{\text{r}}=\sqrt{P_\text{s}} h_{\text{sr}}\bm{x}_{\text{s}}^{\text{i}} + \bm{z}_{\text{r}}^{\text{i}}.
\end{equation}
Here, $\bm{z}_{\text{r}}^{\text{i}}$ is AWGN of the information receiver of the relay with covariance matrix $\left(\sigma_{\text{r}}^{\text{i}}\right)^2\bm{I}_{\frac{(1-\theta)L}{2}}$. Finally, the relay uses the last subslot of length $\frac{1-\theta}{2}L$ cu to amplify and forward the information  signal to the destination. In this way, at the end of Phase 3, the destination's received signal is given by
\begin{equation}\label{eq:bas_chan_it2}
\bm{y}_{\text{d}}=\sqrt{P_{\text{r}}^{\text{TSR}}}h_{\text{rd}}\bm{x}_{\text{r}}^{\text{i}} + \bm{z}_{\text{d}},
\end{equation}
where $\bm{z}_{\text{d}}$ and $P_{\text{r}}^{\text{TSR}}$ are the AWGN at the destination with covariance matrix $\sigma_{\text{d}}^2\bm{I}$ and the relay's transmission power, respectively. Also, in (\ref{eq:bas_chan_it2}), $\bm{x}_{\text{r}}^{\text{i}}=\frac{\bm{y}_{\text{r}}^{\text{i}}}{\sqrt{P_\text{s}g_{\text{sr}}+\left(\sigma^{\text{i}}_\text{r}\right)^2}}$ is the unit-power relay's information signal. Furthermore, we assume that the relay uses all of the harvested energy in Phase 1 to forward the message.  Assuming that the energy required for signal processing at the relay is negligible compared to the energy for data transmission, $P_{\text{r}}^{\text{TSR}}$ is given by
$P^{\text{TSR}}_{\text{r}}=\frac{E^{\text{TSR}}_\text{r}}{\frac{1-\theta}{2}L}=\frac{2\theta}{1-\theta}P_\text{s}g_\text{sr}$.
Substituting (\ref{eq:bas_chan_it1}) in (\ref{eq:bas_chan_it2}) and with some manipulation, the SNR at the destination is obtained as
\begin{equation} \label{eq:snr_TSR}
\gamma_{\text{d,TSR}}=\frac{\frac{P_{\text{s}}g_{\text{sr}}}{\left(\sigma_{\text{r}}^{\text{i}}\right)^2}\frac{P_{\text{r}}^{\text{TSR}}g_{\text{rd}}}{\sigma_{\text{d}}^2}}{\frac{P_{\text{s}}g_{\text{sr}}}{\left(\sigma_{\text{r}}^{\text{i}}\right)^2}+\frac{P_{\text{r}}^{\text{TSR}}g_{\text{rd}}}{\sigma_{\text{d}}^2}+1}\simeq P_{\text{s}}g_{\text{sr}}\frac{2\theta g_{\text{rd}}}{2\theta g_{\text{rd}}+1-\theta},
\end{equation}  
where the last equality is obtained by neglecting $1$ in the denominator. This is in harmony with, e.g., \cite{af1,nasir13} and an appropriate assumption in moderate/high SNRs. Also, with no loss of generality, we assume $\sigma_{\text{r}}^{\text{i}}=\sigma_{\text{d}}=1$.  Since we have set the noise power to $1$, we refer to $P_\text{s}$ (in dB, $10\log_{10} P_\text{s}$) as average SNR. The information signal length for the TSR protocol is given by $L_\text{I,TSR}=\frac{L\left(1-\theta\right)}{2}$ cu. Also, in each communication block, $K$ information nats are encoded into a codeword of length $L_{\text{I,TSR}}$, and the codeword rate is given by $R_{\text{TSR}}\triangleq \frac{K}{L_\text{I,TSR}}=\frac{2K}{L\left(1-\theta\right)}$.
\subsection{Power Splitting Relaying}\label{subsec:PSR}
This protocol comprises two phases of duration $\frac{L}{2}$ cu. In the first phase, the source transmits a unit-variance message $\bm{x}_\text{s}\in \mathbb{C}^{\frac{L}{2}}$ to the relay. Then, the relay splits the source transmitted signal into two streams, one for energy harvesting and one for data reception. Let $\alpha$ denote the power splitting factor. In this way, the received signal at the input of the energy harvester circuit and the information reception circuit can be expressed as
\begin{align}\label{eq:chan_mod_ps1}
\bm{y}_{\text{r}}^{\text{e}}&=\sqrt{\alpha P_\text{s}} h_{\text{sr}}\bm{x}_{\text{s}} + \bm{z}_{\text{r}}^{\text{e}},\\
\bm{y}_{\text{r}}^{\text{i}}&=\sqrt{\left(1-\alpha\right) P_\text{s}} h_{\text{sr}}\bm{x}_{\text{s}} + \bm{z}_{\text{r}}^{\text{i}},\label{eq:chan_mod_ps2}
\end{align}
respectively. Thus, the harvested energy at the end of Phase 1 is given by $E_{\text{r}}^{\text{PSR}}= \alpha \frac{L}{2}g_\text{sr}P_{\text{s}}$. In the second phase, the relay amplifies and forwards the source signal to the destination. Therefore, the received signal at the destination can be written as 
\begin{equation}
\bm{y}_{\text{d}}=\sqrt{P_{\text{r}}^{\text{PSR}}}h_{\text{rd}}\bm{x}_{\text{r}}^{\text{i}} + \bm{z}_{\text{d}},
\end{equation}
where  $\bm{x}_{\text{r}}^{\text{i}}=\frac{\bm{y}_{\text{r}}^{\text{i}}}{\sqrt{\left(1-\alpha\right)P_\text{s}g_{\text{sr}}+\left(\sigma^{\text{i}}_\text{r}\right)^2}}$ is the unit-power relay's information signal. Also, the relay's transmission power is given by $P^{\text{PSR}}_\text{r}=\frac{E_{\text{r}}^{\text{PSR}}}{\frac{L}{2}}=  \alpha  P_{\text{s}} g_\text{sr}$ is the relay's transmission power. With some manipulations, the SNR at the destination is given by
\begin{equation}\label{eq:snr_PSR}
\gamma_{\text{d,PSR}}=\frac{\frac{P_{\text{s}}g_{\text{sr}}}{\left(\sigma_{\text{r}}^{\text{i}}\right)^2}\frac{P^{\text{PSR}}_{\text{r}}g_{\text{rd}}}{\sigma_{\text{d}}^2}}{\frac{P_{\text{s}}g_{\text{sr}}}{\left(\sigma_{\text{r}}^{\text{i}}\right)^2}+\frac{P^{\text{PSR}}_{\text{r}}g_{\text{rd}}}{\sigma_{\text{d}}^2}+1}\simeq \left(1-\alpha\right) P_{\text{s}}g_{\text{sr}}\frac{\alpha g_{\text{rd}}}{\alpha g_{\text{rd}}+1-\alpha},
\end{equation}  
where the approximation is based on the same arguments as in (\ref{eq:snr_TSR}). The information signal length for the PSR protocol is given by $L_\text{I,PSR}=\frac{L}{2}$ cu. In this way, in each communication block, $K$ information nats are encoded into a codeword of length $L_{\text{I,PSR}}$, and the codeword rate is $R_{\text{PSR}}\triangleq \frac{K}{L_\text{I,PSR}}=\frac{2K}{L}$.
\subsection{Review of the results for the finite-blocklength channel coding }
In this part, we review some of the results of \cite{bl_1} on the
achievable rates of finite-block length codes as follows. Define an $ ( L , N, P, \epsilon )$ code as the collection of
\begin{itemize}
\item An encoder $ \Upsilon : \{ 1,\hdots, N \} \mapsto \mathcal{C}^L$ which maps the message $n \in \{ 1,\hdots, N \}$ into a length-$L$ codeword $\bm{x}_n\in
\lbrace{ \bm{x}_1,. . ., \bm{x}_N\rbrace}$ satisfying the power constraint
\begin{equation}
\frac{1}{L}||\bm{x}_j||^2\leq P, \forall j
\end{equation}
\item A decoder $\Lambda$ : $\mathcal{C}^
L \mapsto \{ 1,\hdots, N \}$ satisfying the maximum error probability constraint
\begin{equation}
\max_{\forall j} \ \mathbb{P}\left[\Lambda \left(y(j)\right)\neq j\right]\leq \epsilon,
\end{equation}
with $y(j)$ denoting the channel output induced by the
transmission of codeword $j$.
\end{itemize}
The maximum achievable rate of the code (in nats per channel use (npcu)) is defined as
\begin{equation}\label{eq:max_rate_def}
R_{\text{max}}=\sup \left\lbrace\frac{\log N}{L}\big| \exists ( L , N, P, \epsilon ) \ \text{code} \right\rbrace \quad (\text{npcu}).
\end{equation}
Considering such codes, \cite{bl_1} has recently
presented a very tight approximation for the maximum achievable rate (\ref{eq:max_rate_def}) as
\begin{equation}\label{eq:max_rate_def2}
\begin{aligned}
R_\text{max} \left( L,P,\epsilon \right)&= \sup \left\lbrace R : \mathrm{Pr}\left(\log\left(1+g P\right)< R \right)< \epsilon \right\rbrace\\
&\hspace{2cm}-\mathcal{O}\left(\frac{\log L}{L}\right),  \quad (\text{npcu})
\end{aligned}
\end{equation}
where $g$ is the instantaneous channel gain. Alternatively, the error probability for the transmission with fixed-codeword rate $R$ in a block-fading channel is given by \cite[Eq. 11]{behmain}
\begin{equation}\label{eq:error_main}
\epsilon \simeq \mathbb{E} \left[ Q\left(\frac{\sqrt{L}\left(\log \left(1+\gamma\right)-R\right)}{\sqrt{1-\frac{1}{\left(1+\gamma\right)^2}}}\right)\right],
\end{equation}
where $Q(x)=\int_{x}^{\infty} \frac{1}{\sqrt{2\pi}} e^{-\frac{x^2}{2}} \text{d}x$ denotes the complementary Gaussian cumulative distribution function and the expectation is over the received SNR $\gamma$.
\section{Performance Analysis} \label{sec:non}
According to (\ref{eq:snr_TSR}) and (\ref{eq:snr_PSR}), the SNRs at the destination for the PSR and TSR protocols have the same form
\begin{equation}\label{eq:cdf_koli}
\gamma_{\text{d,A}}=c_{\text{1,A}} P_\text{s}g_{\text{sr}}\frac{c_{\text{2,A}}g_{\text{rd}}}{c_{\text{2,A}} g_{\text{rd}}+1} \quad \text{A} \in \{\text{TSR,PSR}\},
\end{equation}
where
\begin{equation}\label{eq:cT}
c_{\text{1,TSR}}=1,\quad c_{\text{2,TSR}}=\frac{2\theta}{1-\theta},
\end{equation}
for the TSR protocol, and
\begin{equation}\label{eq:cP}
c_{\text{1,PSR}}=1-\alpha,\quad c_{\text{2,PSR}}=\frac{\alpha}{1-\alpha},
\end{equation} 
for the PSR approach. Using (\ref{eq:error_main}), the error probabilities for the TSR and PSR protocols are given by
\begin{equation}\label{eq:err_asli}
\mathcal{P}_{\text{out,A}}=\mathbb{E}\left[Q\left(\frac{\sqrt{L}_{\text{I,A}}\left(\log \left(1+\gamma_{\text{d,A}}\right)-R_{\text{A}}\right)}{\sqrt{1-\frac{1}{\left(1+\gamma_{\text{d,A}}\right)^2}}}\right)\right] ,
\end{equation}
for $\text{A}\in\{\text{TSR},\text{PSR}\}$. Here, in (\ref{eq:err_asli}), expectation is taken with respect to $\gamma_\text{d,A}$ defined in (\ref{eq:cdf_koli}). Also, the throughputs of the TSR and PSR protocols are found as
\begin{equation}
\mathcal{T}_{\text{TSR}}=\frac{R_{\text{TSR}}\left(1-\theta\right)}{2}\left(1-\mathcal{P}_{\text{out},\text{TSR}}\right),
\end{equation}
and
\begin{equation}
\mathcal{T}_{\text{PSR}}=\frac{R_{\text{PSR}}}{2}\left(1-\mathcal{P}_{\text{out},\text{PSR}}\right),
\end{equation}
respectively. In this way, to analyze the network outage probability and throughput, the final step is to derive (\ref{eq:err_asli}). With the considered fading models, the outage probability in (\ref{eq:err_asli}) does not have a closed-form expression. For this reason, Lemmas 2-3 present tight approximations of  $\mathcal{P}_{\text{out,A}}$. Then, in Lemma 4, we derive high-SNR approximation for the outage probability. Let us first derive the cumulative distribution function (CDF) of the received SNR (\ref{eq:cdf_koli}) as follows.
\begin{mylem}
\normalfont
Define the random variable $Z_{\text{A}}\triangleq c_{\text{1,A}}  P_\text{s} X\frac{c_{\text{2,A}} Y}{c_{\text{2,A}} Y+1}$ where $X$ and $Y$ are two random variables following the exponential distribution with mean $1$. Also, $c_{\text{1,A}}$ and $c_{\text{2,A}}$ are constants. Then, the CDF  of $Z_{\text{A}}$ is given by
\end{mylem}
\begin{equation}\label{eq:CDF_main}
F_{Z_{\text{A}}}\left(t\right)=1-e^{\left(-\frac{t}{c_{\text{1,A}} P_\text{s}}\right)}\sqrt{\frac{4t}{ c_{\text{1,A}} P_\text{s} c_{\text{2,A}} }}\mathcal{K}_{1}\left(\sqrt{\frac{4t}{c_{\text{1,A}} P_\text{s} c_{\text{2,A}}}}\right),
\end{equation}
where $\mathcal{K}_1\left(\cdot\right)$ is the first-order modified Bessel function of the second kind.
\begin{proof}
According to the definition of CDF, we have
\begin{equation} \label{eq:cdf_no}
\begin{aligned}
&F_{Z_{\text{A}}}\left(t\right)
=\int\limits_{0}^{\infty}f_{Y}\left(y\right)\mathbb{P}\left[X\leq \frac{t\left(1+c_{\text{2,A}} y\right)}{c_{\text{1,A}} P_\text{s} c_{\text{2,A}} y}\right]\text{d}y\\
&=1-\int\limits_{0}^{\infty} e^{-y}e^{\left(-\frac{t\left(1+c_{\text{2,A}} y\right)}{c_{\text{1,A}} P_\text{s} c_{\text{2,A}}y}\right)}\text{d}y\\
&=1-e^{\left(-\frac{t}{ c_{\text{1,A}} P_\text{s}}\right)}\frac{t}{c_{\text{1,A}} P_\text{s} c_{\text{2,A}}}\int\limits_0^\infty e^{\left(-\frac{tw}{ c_{\text{2,A}} c_{\text{1,A}} P_\text{s}}-\frac{1}{w}\right)}\text{d}w.
\end{aligned}
\end{equation}
Here, the last equality is obtained using the change of variable $w=\frac{c_{\text{1,A}} P_\text{s} c_{\text{2,A}} y}{t}$. Then, by using the identity in \cite[Eq. 3.324.1]{tabebesel}, the final result in (\ref{eq:CDF_main}) is obtained.
\end{proof}
\begin{mylem} \label{thr:min}
\normalfont
For every $M_1,M_2,M_3>0,$ the probability $\mathcal{P}_{\text{out,A}}$ is 
approximated by $\tilde{\mathcal{P}}_{\text{out,A}}$ for $\text{A}\in\{\text{TSR,PSR}\}$ in (\ref{eq:aprrox_1}) as shown on the top of the next page. In (\ref{eq:aprrox_1}),
\newcounter{MYtempeqncnt}
\begin{figure*}[!t]
% ensure that we have normalsize text
\normalsize
% Store the current equation number.
\setcounter{MYtempeqncnt}{\value{equation}}
% Set the equation number to one less than the one
% desired for the first equation here.
% The value here will have to changed if equations
% are added or removed prior to the place these
% equations are referenced in the main text.
\begin{equation} \label{eq:aprrox_1}
\small
\begin{aligned}
&\tilde{\mathcal{P}}_{\text{out,A}}=\\
& \sum\limits_{m=0}^{M_1}\sum\limits_{j=0}^{M_2} \frac{(-1)^{m}(-c_{\text{2,A}})^j\mathcal{U}_{m,j}\left[\mathcal{W}_j\left(\infty\right)-\mathcal{W}_j\left(-\lambda\right) \right]}{c^m_{\text{2,A}} m! j! \left(\sqrt{L_{\text{I,A}}}c_{\text{1,A}} c_{\text{2,A}} P_{\text{s}} e^{-R_{\text{A}}}\right)^{j+1}}+ 
e^{-\frac{1}{c_{\text{2,A}}}} \Bigg[ Q\left(-\sqrt{L_{\text{I,A}}}R_{\text{A}}\right) - e^{\frac{1}{c_{\text{1,A}} P_\text{s}}} \sum_{n=0}^{M_3}\frac{(-1)^n}{n!} \left(\frac{e^{R_{\text{A}}}}{c_{\text{1,A}} P_\text{s}}\right)^n  e^{\frac{n^2}{2L_{\text{I,A}}}} Q\left(\frac{-LR_{\text{A}}-n}{\sqrt{L_{\text{I,A}}}}\right) \Bigg]
 \Bigg]
\end{aligned}
\end{equation}
% Restore the current equation number.
\setcounter{equation}{\value{MYtempeqncnt}+1}
% The IEEE uses as a separator
\hrulefill
% The spacer can be tweaked to stop underfull vboxes.
\vspace*{0pt}
\end{figure*}
\begin{equation}\label{eq:umj}
\mathcal{U}_{m,j}=\begin{cases}\vspace{0.3cm}
E_{m+1-j}\left(\frac{1-e^{-R_{\text{A}}}}{c_{\text{1,A}} P_{\text{s}} e^{-R_{\text{A}}}}\right) &\text{ $m+1>j$}\\ \vspace{0.3cm}
\frac{c_{\text{1,A}} P_{\text{s}} e^{-R_{\text{A}}}}{1-e^{-R_{\text{A}}}}e^{\left(-\frac{1-e^{-R_{\text{A}}}}{c_{\text{1,A}} P_{\text{s}} e^{-R_{\text{A}}}}\right)}&\text{ $m+1=j$}\\
\frac{\Gamma\left(m+2-j,\frac{1-e^{-R_{\text{A}}}}{c_{\text{1,A}} P_{\text{s}} e^{-R_{\text{A}}}}\right)}{\left(\frac{1-e^{-R_{\text{A}}}}{c_{\text{1,A}} P_{\text{s}} e^{-R_{\text{A}}}}\right)^{m+2-j}} &\text{ $m+1<j$}
\end{cases},
\end{equation}
and
\begin{equation}\label{eq:defw}
\small
\mathcal{W}_j\left(x\right)=\frac{x^{j+1}}{j+1}Q\left(x\right)+\frac{\left(\sqrt{2}\right)^{j-1}}{\sqrt{\pi}\left(j+1\right)}\gamma\left(\frac{j+2}{2},\frac{x^2}{2}\right),
\end{equation}
where $\gamma(s,a)$ and $\Gamma(s,a)$ denote the lower and the upper incomplete Gamma functions, respectively.
\end{mylem}
\begin{proof}
At medium/high SNR, $\frac{\log(1+x)-r}{\sqrt{1-\frac{1}{(1+x)^2}}}$ is tightly approximated by  $\log(1+x)-r$. In this way, the error probability (\ref{eq:err_asli}) is given by
\begin{equation}\label{eq:lem2_main}
\small
\begin{aligned}
&\mathcal{P}_{\text{out,A}}
\simeq\mathbb{E}\left[Q\left(\sqrt{L_{\text{I,A}}}\left(\log\left(1+c_{\text{1,A}}P_{\text{s}}g_{\text{sr}}\frac{c_{\text{2,A}} g_{\text{rd}}}{1+c_{\text{2,A}} g_{\text{rd}}}\right)-R_{\text{A}}\right)\right)\right]\\
&\stackrel{\text{(a)}}{\simeq} \int\limits_{0}^{\infty}\int\limits_{0}^{\infty} Q\left(\sqrt{L_{\text{I,A}}}\left(\log\left(1+c_{\text{1,A}} P_\text{s} x \min\{1,c_{\text{2,A}} y\}\right)-R_{\text{A}}\right)\right)\\
& \hspace{5cm}\times
e^{-x}e^{-y} \text{d}x\text{d}y\\
&=\underbrace{\int\limits_{0}^{\frac{1}{c_{\text{2,A}}}} \int\limits_{0}^{\infty} Q\left(\sqrt{L_{\text{I,A}}}\left(\log\left(1+c_1 P_\text{s} c_{\text{2,A}} xy\right)-R_{\text{A}}\right)\right)e^{-x}e^{-y}\text{d}x\text{d}y}_{\mathcal{I}_1}
\\
&+\underbrace{ \int\limits_{\frac{1}{c_{\text{2,A}}}}^{\infty} \int\limits_{0}^{\infty} Q\left(\sqrt{L_{\text{I,A}}}\left(\log\left(1+c_1 P_\text{s} x\right)-R_{\text{A}}\right)\right)e^{-x}e^{-y}\text{d}x\text{d}y}_{\mathcal{I}_2}.
\end{aligned}
\end{equation}
Here, in $(a)$, we use $\frac{c_{\text{2,A}} g_{\text{rd}}}{1+c_{\text{2,A}} g_{\text{rd}}} \leq \min\{c_{\text{2,A}} g_{\text{rd}},1\}$ to approximate $\frac{c_{\text{2,A}} g_{\text{rd}}}{1+c_{\text{2,A}} g_{\text{rd}}}$. To find $\mathcal{I}_1$, we use the change of variable $w = xy$ to write 
\begin{equation}\nonumber
\small
\begin{aligned}
&\mathcal{I}_1=
 \int\limits_{0}^{\infty} \int\limits_{c_{\text{2,A}}w}^{\infty} Q(\sqrt{L_{\text{I,A}}} \left(\log\left(1+c_{\text{1,A} }c_{\text{2,A}} P_{\text{s}} w\right)-R_{\text{A}}\right))e^{-x}e^{-\frac{w}{x}} \frac{1}{x}\text{d}x\text{d}w\\
&\stackrel{\text{(b)}}{=}  \sum\limits_{m=0}^{\infty} \frac{(-1)^m}{m!}\int\limits_{0}^{\infty}Q\left(\sqrt{L_{\text{I,A}}} \left(\log\left(1+c_{\text{1,A}} c_{\text{2,A}} P_{\text{s}} w\right)-R_{\text{A}}\right)\right)\\
& \hspace{5cm} \times \int\limits_{c_{\text{2,A}} w}^{\infty} e^{-x} \left(\frac{w}{x}\right)^m \frac{1}{x}\text{d}x  \text{d}w
 \end{aligned}
\end{equation}
\begin{equation}\label{eq:I1_main}
\small
\begin{aligned}
&\stackrel{\text{(c)}}{=} \sum\limits_{m=0}^{\infty}\frac{(-1)^m}{c^m_{\text{2,A}} m!}\int\limits_{0}^{\infty}Q\left(\sqrt{L_{\text{I,A}}} \left(\log\left(1+c_{\text{1,A}} c_{\text{2,A}} P_{\text{s}} w\right)-R_{\text{A}}\right)\right)\\
&\hspace{5cm} \times E_{m+1}\left(c_{\text{2,A}} w\right) \text{d}w\\
&\stackrel{\text{(d)}}{\simeq}  \sum\limits_{m=0}^{\infty} \frac{(-1)^m}{c^m_{\text{2,A}} m!} \int\limits_{0}^{\infty} Q\left(\delta w- \lambda\right)  E_{m+1}\left(c_{\text{2,A}} w\right) \text{d}w\\
&\stackrel{\text{(e)}}{=}\sum\limits_{m=0}^{\infty} \frac{(-1)^m}{c^m_{\text{2,A}} m!\delta}\int\limits_{1}^{\infty} \frac{e^{-\frac{c_{\text{2,A}}\lambda t}{\delta}}}{t^{m+1}}\left[ \int\limits_{-\lambda}^{\infty}  Q\left(u\right) e^{- c_{\text{2,A}} \frac{u}{\delta}t} \text{d}u\right]\text{d}t\\
&=\sum\limits_{m=0}^{\infty}\sum\limits_{j=0}^{\infty} \frac{(-1)^{m}(-c_{\text{2,A}})^j}{c_{\text{2,A}}^m m! j! \delta^{j+1}}\underbrace{\int\limits_{1}^{\infty} \frac{e^{-\frac{\lambda c_{\text{2,A}} t}{ \delta}}}{t^{m+1}}t^j\text{d}t}_{\mathcal{U}_{m,j}}\underbrace{\int\limits_{-\lambda}^{\infty}Q\left(u\right) u^j\text{d}u}_{\mathcal{W}_j\left(\infty\right)-\mathcal{W}_j\left(-\lambda\right)} \\
&=\sum\limits_{m=0}^{\infty}\sum\limits_{j=0}^{\infty} \frac{(-1)^{m}(-c_{\text{2,A}})^j}{c_{\text{2,A}}^m m! j! \delta^{j+1}}\mathcal{U}_{m,j}\left[\mathcal{W}_j\left(\infty\right)-\mathcal{W}_j\left(-\lambda\right)\right],
\end{aligned}
\end{equation}
with $\mathcal{U}_{m,j}$ and $\mathcal{W}_j\left(x\right)$ defined in (\ref{eq:umj}) and 
(\ref{eq:defw}), respectively. Here, $\delta\triangleq \sqrt{L_{\text{I,A}}}c_{\text{1,A}} c_{\text{2,A}} P_{\text{s}} e^{-R_{\text{A}}}$ and $\lambda\triangleq\sqrt{L_{\text{I,A}}}\left(1-e^{-R_{\text{A}}}\right)$. Then, in $(b)$, we have used the Taylor series representation of $e^{-\frac{w}{v}}$. Also, $(c)$ is found by the definition of generalized exponential integral in \cite[Eq. 8.211.1]{tabebesel}, $E_{m}\left(x\right)=\int_{1}^{\infty}\frac{e^{-xt}}{t^m}\text{d}t.$ Also, considering the fact that $Q(x)$ is a decreasing function and $\log(1+c_{\text{1,A}}c_{\text{2,A}}Pw)-R_{\text{A}}$ is a concave function in $w$, in step $(d)$, $\sqrt{L_{\text{I,A}}}\left(\log(1+c_{\text{1,A}} c_{\text{2,A}}P_{\text{s}}w)-R_{\text{A}}\right)$ is approximated by its first order Taylor expansion at $w=\frac{e^{R_{\text{A}}}-1}{c_{\text{1,A}} c_{\text{2,A}} P_{\text{s}}}$. Then, $(e)$ follows from the change of variable $u = \delta w- \lambda$ and using the integral representation of $E_{m+1}\left(\cdot\right)$ given in  \cite[Eq. 8.211.1]{tabebesel}. Finally, the last equality is based on \cite[Eq. 4.1.17]{ng1969table} and some manipulations.
\par Then, $\mathcal{I}_2$ in (\ref{eq:lem2_main}) is given by
\begin{equation}\label{eq:I2}
\small
\begin{aligned}
\mathcal{I}_2&=e^{-\frac{1}{c_{\text{2,A}}}}\int\limits_{0}^{\infty}Q\left(\sqrt{L_{\text{I,A}}}\left(\log\left(1+P_\text{s} c_{\text{1,A}} x\right)-R_{\text{A}}\right)\right)e^{-x}\text{d}x\\
&\stackrel{\text{(f)}}{=} \frac{e^{R_{\text{A}}}}{c_{\text{1,A}} P_\text{s}} e^{-\frac{1}{c_{\text{2,A}}}}e^{\frac{1}{c_{\text{1,A}} P_\text{s}}}\int\limits_{e^{-R_{\text{A}}}}^{\infty}Q\left(\sqrt{L_{\text{I,A}}}\log\left(u\right)\right)e^{-\frac{ue^{R_{\text{A}}}}{c_{\text{1,A}} P_\text{s}}}\text{d}x\\
&\stackrel{\text{(g)}}{=}  e^{-\frac{1}{c_{\text{2,A}}}} \bigg[ Q\left(-\sqrt{L_{\text{I,A}}}R_{\text{A}}\right) - \sqrt{\frac{L_{\text{I,A}}}{2\pi}}e^{\frac{1}{c_{\text{1,A}} P_\text{s}}} \sum_{n=0}^{\infty}\frac{(-1)^n}{n!} \\ 
& \hspace{2cm} \times \left(\frac{e^{R_{\text{A}}}}{c_{\text{1,A}} P_\text{s}}\right)^n \int\limits_{e^{-R_{\text{A}}}}^{\infty} e^{-\frac{L_{\text{I,A}}\left(\log\left(u\right)\right)^2}{2}}  u^{n-1}   \text{d}u  \bigg]\\
&=
 e^{-\frac{1}{c_{\text{2,A}}}} \Bigg[ Q\left(-\sqrt{L_{\text{I,A}}}R_{\text{A}}\right) - e^{\frac{1}{c_{\text{1,A}} P_\text{s}}} \sum_{n=0}^{\infty}\frac{(-1)^n}{n!} \left(\frac{e^{R_{\text{A}}}}{c_{\text{1,A}} P_\text{s}}\right)^n \\
& \hspace{4cm}  \times e^{\frac{n^2}{2L_{\text{I,A}}}} Q\left(\frac{-L_{\text{I,A}}R_{\text{A}}-n}{\sqrt{L_{\text{I,A}}}}\right) \Bigg],
\end{aligned}
\end{equation}
where, $(f)$ is obtained by using the change of variable $u = \left(1+P_{\text{s}}c_{\text{1,A}}x\right)e^{-R_{\text{A}}}$. Also, $(g)$ follows from partial integration and the Taylor series expansion of $e^{-\frac{ue^{R_{\text{A}}}}{c_{\text{1,A}} P_\text{s}}}$.\par
Finally, substituting (\ref{eq:I1_main}) and (\ref{eq:I2}) into (\ref{eq:lem2_main}) and considering the finite number of summations, the expression in (\ref{eq:aprrox_1}) is obtained where $M_1$, $M_2$ and $M_3$ denote the number of summation terms in (\ref{eq:I1_main}) and (\ref{eq:I2}), respectively. Note that the tightness of the approximations increases by $M_1,M_2,M_3.$ However, as also seen in Figs. 1-3, the approximation result of (\ref{eq:aprrox_1}) matches tightly with the numerical results even by $M_1=2, M_2=2, M_3=1.$ Moreover, as will be shown in Figs. 1 and 3, the approximation in (\ref{eq:aprrox_1}) is tight in the medium to high SNR regime, while its tightness decreases at low SNRs comparatively.
\end{proof}
\begin{mylem}\label{lem:beh}
\normalfont
The error probability (\ref{eq:err_asli}) is approximated by
\end{mylem}
\begin{equation}\label{eq:beh_out}
\small
\begin{aligned}
&\mathcal{P}_{\text{out,A}}\simeq1-\frac{\beta\sqrt{L_{\text{I,A}}}  P_{\text{s}} c_{\text{1,A}} c_{\text{2,A}} }{2}
\sum_{n=0}^{M_4}\frac{(-c_{\text{2,A}})^n}{4^n n!}\\
&  \times\Bigg[ \mathcal{Y}_n\left(\sqrt{\frac{4\left(\alpha+\frac{1}{2\beta \sqrt{L_{\text{I,A}}}}\right)}{  P_{\text{s}} c_{\text{1,A}} c_{\text{2,A}}}}\right)-\mathcal{Y}_n\left(\sqrt{\frac{4\left(\alpha-\frac{1}{2 \beta  \sqrt{L_{\text{I,A}}}}\right)}{  P_{\text{s}} c_{\text{1,A}} c_{\text{2,A}}}}\right)\Bigg],
\end{aligned}
\end{equation}
with
\begin{equation*}
\mathcal{Y}_n\left(x\right)
=2^{2n+1} \mathcal{G}_{1,3}^{2,1}\left(\frac{x}{2},\frac{1}{2}\Bigg| 
  \begin{array}{ll}
     \ \ \ \ \ \  \ \  1   \\
    n+1,n+2,0 
  \end{array}\right).
\end{equation*}
Here, we define $\alpha=e^{R_{\text{A}}}-1$, $\beta=\sqrt{\frac{1}{2\pi\left(e^{
2R_{\text{A}}}-1\right)}}$ and $M_4\ge0$ is an arbitrary number which is selected such that (\ref{eq:beh_out}) provides an appropriate approximation for the outage probability. Also, in (\ref{eq:beh_out}), $\mathcal{G}_{a_1,a_2}^{b_1,b_2}\left(\cdot,\cdot\Bigg| 
  \begin{array}{ll}
     \ \  \cdot   \\
    \cdot,\cdot,\cdot 
  \end{array}\right)$ is the generalized Meijer G-function \cite[Eq. 9.301]{tabebesel}. 
\begin{proof}
In order to prove Lemma \ref{lem:beh}, we use a linear approximation $Q\left(\frac{\sqrt{L_{\text{I,A}}}\left(\log \left(1+x\right)-R_{\text{A}}\right)}{\sqrt{1-\frac{1}{\left(1+x\right)^2}}}\right)$ $\simeq$ $\mathcal{J}\left(x\right)$ with
\begin{equation} \label{eq:app_beh}
\small
\begin{aligned}
\mathcal{J}\left(x\right)=\begin{cases}\vspace{0.15cm}
1 &\text{ $x\leq \alpha-\frac{1}{2\sqrt{L_{\text{I,A}}}\beta}$}\\ 
\frac{1}{2}-\beta\sqrt{L_{\text{I,A}}}\left(x-\alpha\right)&\text{ $\alpha-\frac{1}{2\sqrt{L_{\text{I,A}}}\beta}\leq x \leq \alpha+\frac{1}{2\sqrt{L_{\text{I,A}}}\beta} $}\\
0 &\text{ $x\geq \alpha+\frac{1}{2\sqrt{L_{\text{I,A}}}\beta}$}
\end{cases},
\end{aligned}
\end{equation}
where $\alpha=e^{R_{\text{A}}}-1$ and $\beta=\frac{\partial\left(  Q\left(\frac{\sqrt{L_{\text{I,A}}}\left(\log \left(1+x\right)-R_{\text{A}}\right)}{\sqrt{1-\frac{1}{\left(1+x\right)^2}}}\right)\right)}{L_{\text{I,A}} \partial x}=\sqrt{\frac{1}{2\pi\left(e^{2R_{\text{A}}}-1\right)}}$. Hence,  (\ref{eq:err_asli}) is approximated by
\begin{equation}\nonumber
\small
\begin{aligned}
&\tilde{\mathcal{P}}_{\text{out,A}}= \int\limits_{0}^{\infty} \mathcal{J}\left(x\right)f_{\gamma_{\text{d}}}\left(x\right)\text{d}x\stackrel{\text{(h)}}{=}\beta\sqrt{L_{\text{I,A}}}\int\limits_{\alpha-\frac{1}{2\beta\sqrt{L_{\text{I,A}}}}}^{\alpha+\frac{1}{2\beta\sqrt{L_{\text{I,A}}}}}F_{\gamma_{\text{d}}}\left(x\right)\text{d}x
\end{aligned}
\end{equation}
\begin{equation}\label{eq:prob_beh}
\small
\begin{aligned}
&\stackrel{\text{(i)}}{=}1-\frac{\beta\sqrt{L_{\text{I,A}}}  P_{\text{s}} c_{\text{1,A}} c_{\text{2,A}} }{2}\int\limits_{\sqrt{\frac{4\left(\alpha-\frac{1}{2\beta\sqrt{L_{\text{I,A}}}}\right)}{ P_{\text{s}} c_{\text{1,A}} c_{\text{2,A}}}}}^{\sqrt{\frac{4\left(\alpha+\frac{1}{2\beta\sqrt{L_{\text{I,A}}}}\right)}{ P_{\text{s}} c_{\text{1,A}} c_{\text{2,A}}}}}
e^{\left(-\frac{c_{\text{2,A}}}{4} u^2\right)}u^2\mathcal{K}_{1}\left(u\right)\text{d}u\\
&=1-\frac{\beta \sqrt{L_{\text{I,A}}} P_{\text{s}} c_{\text{1,A}} c_{\text{2,A}} }{2}
\sum_{n=0}^{\infty}\frac{(-c_{\text{2,A}})^n}{4^n n!}\\
&\hspace{3cm} \times \int\limits_{\sqrt{\frac{4\left(\alpha-\frac{1}{2\beta\sqrt{L_{\text{I,A}}}}\right)}{ P_{\text{s}} c_{\text{1,A}} c_{\text{2,A}}}}}^{\sqrt{\frac{4\left(\alpha+\frac{1}{2\beta\sqrt{L_{\text{I,A}}}}\right)}{ P_{\text{s}} c_{\text{1,A}} c_{\text{2,A}}}}} u^{2n+2}\mathcal{K}_{1}\left(u\right)\text{d}u,
\end{aligned}
\end{equation}
where $(h)$ is obtained by partial integration and some manipulations. Then, $(i)$ is found by using (\ref{eq:CDF_main}) and the change of variable $u=\sqrt{\frac{4t}{c_{\text{1,A}} P_\text{s} c_{\text{2,A}}}}$. Also the last equality comes from the Taylor series expansion of $e^{\left(-\frac{c_{\text{2,A}} }{4}u^2\right)}$. Finally, using the representation of $\mathcal{K}_1\left(x\right)$ function in terms of Meijer G-function  \cite[Eq. 9.34.3]{tabebesel}
 $\mathcal{K}_{1}\left(x\right)=\frac{1}{2}\mathcal{G}_{0,2}^{2,0}\left(\frac{x}{2},\frac{1}{2}\Bigg|      
    \frac{n}{2},\frac{-n}{2} 
 \right)$
and considering the finite number of summations which is denoted by $M_4$ the expression in (\ref{eq:beh_out}) is obtained. 
\end{proof}
As will be shown in Section \ref{sec:nemo}, for a broad range of parameter settings, very accurate approximation of the outage probability is obtained by few summations terms in (\ref{eq:beh_out}). For instance, with the parameter settings of Figs. 1-3, very tight approximations of the outage probability are achieved by setting $M_4=2$ in (28).  
\begin{mylem}
\normalfont
At high SNRs, the outage probability (\ref{eq:err_asli}) is approximated by 
\begin{equation} \label{eq:aprrox_hsnr}
\begin{aligned}
&\lim\limits_{P_{\text{s}}\to \infty}\mathcal{P}_{\text{out,A}}\simeq \frac{\alpha\left(\log\left(4P_{\text{s}}\right)+c_{\text{2,A}}-\frac{1}{2}\right)}{c_{\text{1,A}} c_{\text{2,A}} P_{\text{s}}}-\frac{\beta }{2 c_{\text{1,A}} c_{\text{2,A}} P_{\text{s}}}\\
&\hspace{4cm}\times\underbrace{ \sqrt{L_{\text{I,A}}}\left[V\left(T_1\right)-V\left(T_2\right) \right]}_{U\left(L_{\text{I,A}}\right)},
\end{aligned}
\end{equation}
where $V\left(x\right)\triangleq x^2\log\left(\frac{4x}{c_1c_{\text{2,A}}}\right) $, $T_{1}=\alpha+\frac{1}{2\beta\sqrt{L_{\text{I,A}}}}$ and $T_{2}=\alpha-\frac{1}{2\beta\sqrt{L_{\text{I,A}}}}$.
\end{mylem}
\begin{proof}
We use the approximation approach presented in Lemma \ref{lem:beh}, to find the high SNR approximation. In particular, considering step $(i)$ in (\ref{eq:prob_beh}), we have
\begin{equation}
\small
\begin{aligned}
&\mathcal{P}_{\text{out}}\simeq 1-\frac{\beta \sqrt{L_{\text{I,A}}}  P_{\text{s}} c_{\text{1,A}} c_{\text{2,A}} }{2}\int\limits_{\sqrt{\frac{4T_2}{ P_{\text{s}} c_{\text{1,A}} c_{\text{2,A}}}}}^{\sqrt{\frac{4T_1}{ P_{\text{s}} c_{\text{1,A}} c_{\text{2,A}}}}}
e^{-\frac{c_{\text{2,A}}}{4} u^2}u^2\mathcal{K}_{1}\left(u\right)\text{d}u\\
&\simeq 1-\frac{\beta \sqrt{L_{\text{I,A}}}   c_{\text{1,A}} c_{\text{2,A}} }{4} \int\limits_{\frac{4 T_2}{ c_{\text{1,A}} c_{\text{2,A}}}}^{\frac{4T_1}{  c_{\text{1,A}} c_{\text{2,A}}}} \left(1-\frac{c_{\text{2,A}} v}{4P_{\text{s}}}+\frac{v}{4P_{\text{s}}}\log\left(\frac{v}{4P_{\text{s}}}\right)\right)\text{d}v.\\
\end{aligned}
\end{equation}
Here, we use the tight of approximation $t\mathcal{K}_1(t)\simeq 1+\frac{t^2}{2}\log\left(\frac{t}{2}\right)$ for $t\to 0$\cite[Eq. 24]{poor} and the first order Taylor series expansion of $e^{-\frac{c_{\text{2,A}} v}{4P_{\text{s}}}}$. Finally, Considering high-order terms and with some manipulations, the expression in (\ref{eq:aprrox_hsnr}) is obtained.
\end{proof}
Equation (\ref{eq:aprrox_hsnr}) implies that the outage probability decreases linearly with $\frac{\log P}{P}$ at high SNR. This is substantially worse than fixed-energy supply relay networks, where the outage probability decreases linearly with $\frac{1}{P}$\cite{rel_snr}. Moreover, $U\left(L_{\text{I,A}}\right)$ in (\ref{eq:aprrox_hsnr}) characterizes the impact of the codeword length on the outage probability at high SNR. In this way, it can be shown that
\begin{equation}
\small
\lim\limits_{L_\text{I,A}\to \infty} U\left(L_\text{I,A}\right)=\alpha+\frac{2\alpha^2}{\beta}\log\left(\frac{4\alpha}{c_{\text{1,A}} c_{\text{2,A}}}\right).
\end{equation}
Hence, considering infinitely long codes, i.e., $L_{\text{I,A}}\to \infty$, the outage probability is given by
\begin{equation}\label{eq:hsnr_inf}
\small
\mathcal{P}^{L_\text{I}\to \infty}_{\text{out,A},P_{\text{s}}\to \infty}=\frac{\alpha  }{c_{\text{1,A}}  P_{\text{s}}}+\frac{\alpha\left(\log\left(4P_{\text{s}}\right)-\frac{1}{2}\right)}{c_{\text{1,A}} c_{\text{2,A}} P_{\text{s}}}-\frac{\beta \alpha+2\alpha^2\log\left(\frac{4\alpha}{c_{\text{1,A}} c_{\text{2,A}}}\right) }{2 c_{\text{1,A}} c_{\text{2,A}} P_{\text{s}}}.
\end{equation}
In Section \ref{sec:nemo}, we validate the accuracy
of the approximations proposed in Lemma 4 and Eq. (\ref{eq:hsnr_inf}) by comparing
them with the corresponding exact values that can be
evaluated numerically.
\section{Results}\label{sec:nemo}
In Fig. \ref{fig:err_P}, setting $L=1000$ and $R_{\text{TSR}}=R_{\text{PSR}}=2$ npcu, we study the accuracy of the analytical results presented in Lemmas 2-4. In particular, to obtain the results of Lemma 2 and Lemma 3, we have set $(M_1,M_2,M_3)=(2,2,1)$ and $M_4=2$, respectively, which determine the number of summation terms in (\ref{eq:aprrox_1}) and (\ref{eq:beh_out}).  As can be seen in Fig. \ref{fig:err_P}, the approximation approach of Lemma 2 matches the numerical results in the moderate/high SNR regime. However, this approach loses its tightness at low SNR. Also, the approximation scheme of Lemma 3 is tight in a broad range of SNR values. Moreover, with the parameter settings of the figure, the result of Lemma 3 provides better approximation than the approach of Lemma 2. Finally, the figure compares the system outage probability for the TSR and PSR protocols and, lower outage probability is achieved by the TSR protocol. Also, in harmony with Lemma 4, the outage probabilities of both relaying protocols, i.e. the PSR and TSR protocols, decrease linearly with $\frac{\log \text{SNR}}{\text{SNR}}$ at the high SNR regime. \par
Fig. \ref{fig:thr_L} investigates the impact of the total number of channel uses per slot, i.e., $L$, on the throughput of the PSR and TSR protocols. Here, the results are obtained by setting $\text{SNR}=30$ dB and $K=800$ nats. Fig. \ref{fig:thr_L} indicates that the results of the approximation approach of Lemmas 2 and 3 are tight for $L\geq 200$. For a given number of information nats $K$, there is an optimal value of $L$ for which the throughput is maximized (Fig. \ref{fig:thr_L}). Intuitively, this is because for small values of $L$ the code rate is high and outage occurs with high probability. For large $L$, on the other hand, although the message is correctly decoded by the destination with high probability, low throughput is observed due to low code rate. Thus, there is a tradeoff and the maximum throughput is achieved by a finite value of the codeword length. \par
Fig. \ref{fig:thr_ceof} shows the impact of the power splitting and the time sharing factors on the throughput of the PSR and TSR protocols. For both PSR and TSR protocols, the approximation approach of Lemma 3 is very tight such that the plots of
the numerical results and approximation approach of
Lemma 3 are superimposed. On the other hand, for small values of $\theta$ and $\alpha$, where the relay's transmission power is low, the approximation approach of Lemma 2 is loose comparatively, and, for this reason, the results of the approximation approach of Lemma 2 are plotted for its range of interest. In both the PSR and the TSR protocols, for small values of $\alpha$ and $\theta$, the relay cannot harvest enough energy and, consequently, low throughput is observed. On the other hand, as seen in Fig. \ref{fig:thr_ceof}a, for the TSR protocol, as $\theta$ increases, the codeword rate increases and the length of the information signal decreases. Thus, the outage probability remarkably increases and the throughput decreases for large values of $\theta$. In the PSR scheme, on the other hand, as $\alpha$ increases more energy is devoted to the energy harvesting and less energy is allocated to the information signal (Fig. \ref{fig:thr_ceof}b). Consequently, for large values of $\alpha,$ the throughput decreases with $\alpha$ remarkably. Therefore, there exist optimal values for $\alpha$ and $\theta$, in terms of throughput. Additionally, with a fixed SNR, as the number of information nats, i.e. $K$, increases, the optimal time sharing (resp. power splitting) factor, for the TSR (resp. PSR) protocol increases.\par
Then, setting $R_{\text{TSR}}=R_{\text{PSR}}=1$ npcu, Fig. \ref{fig:thr_dif} compares the throughput of the finite blocklength codes and the throughput with asymptotically long codewords. Also, in Fig. \ref{fig:thr_dif}, we compare the approximation approach of (\ref{eq:hsnr_inf}) with the numerical results. Fig. \ref{fig:thr_dif} shows that the high-SNR approximation approach of (\ref{eq:hsnr_inf}) tends towards the exact results as the SNR increases, as expected. Finally, since the throughput degradation that results from using short packets is negligible, the throughput achieved with long codewords can be approached with codewords of finite length.
\begin{figure}
\centering     %%% not \center
\includegraphics[width=3.4in, height=1.7in]{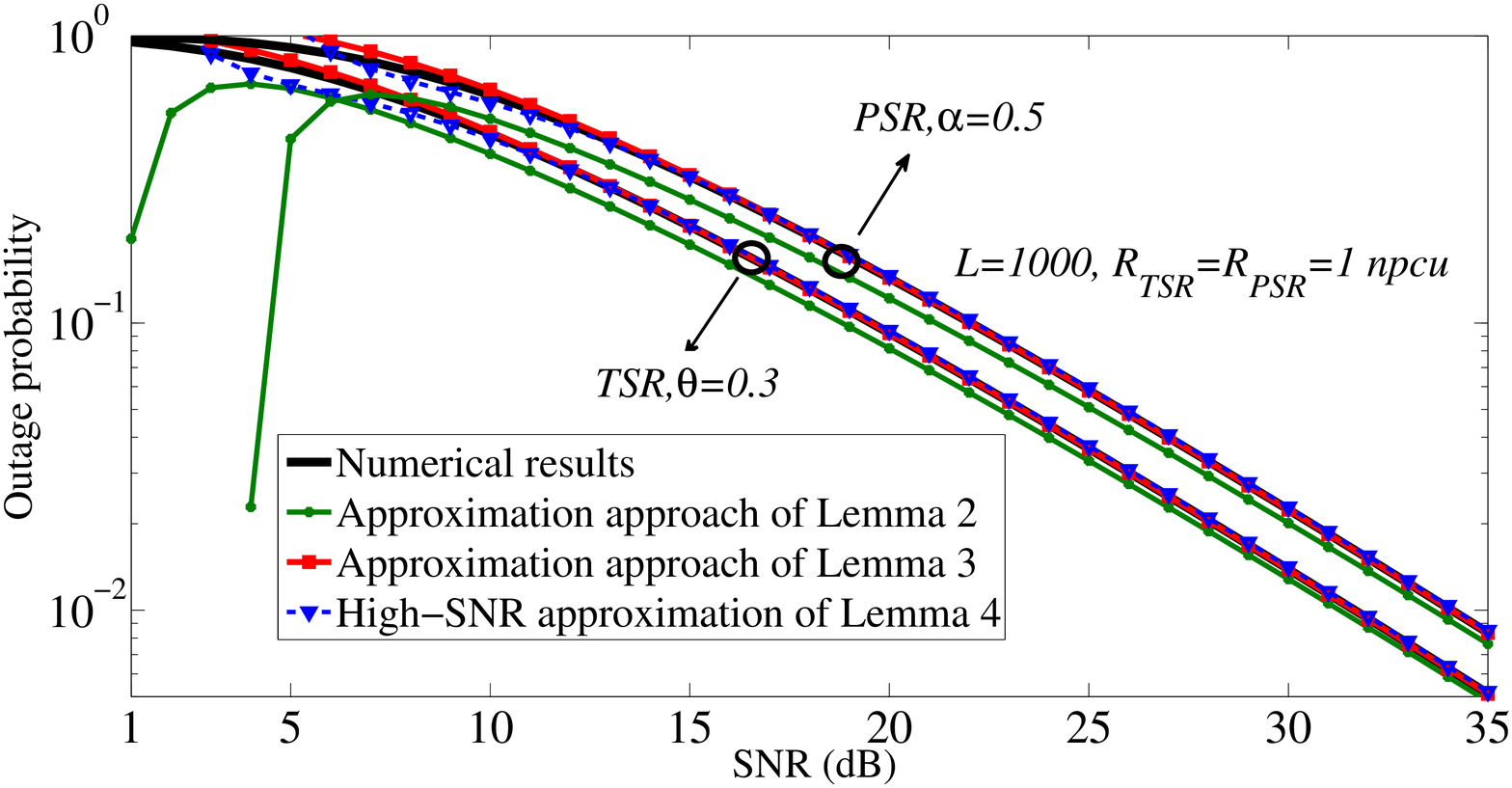}
\caption{Outage probability vs. SNR. $R_{\text{TSR}}=R_{\text{PSR}}=2$ npcu, $L=1000$.}
\label{fig:err_P}
\centering     %%% not \center
\includegraphics[width=3.4in, height=1.60in]{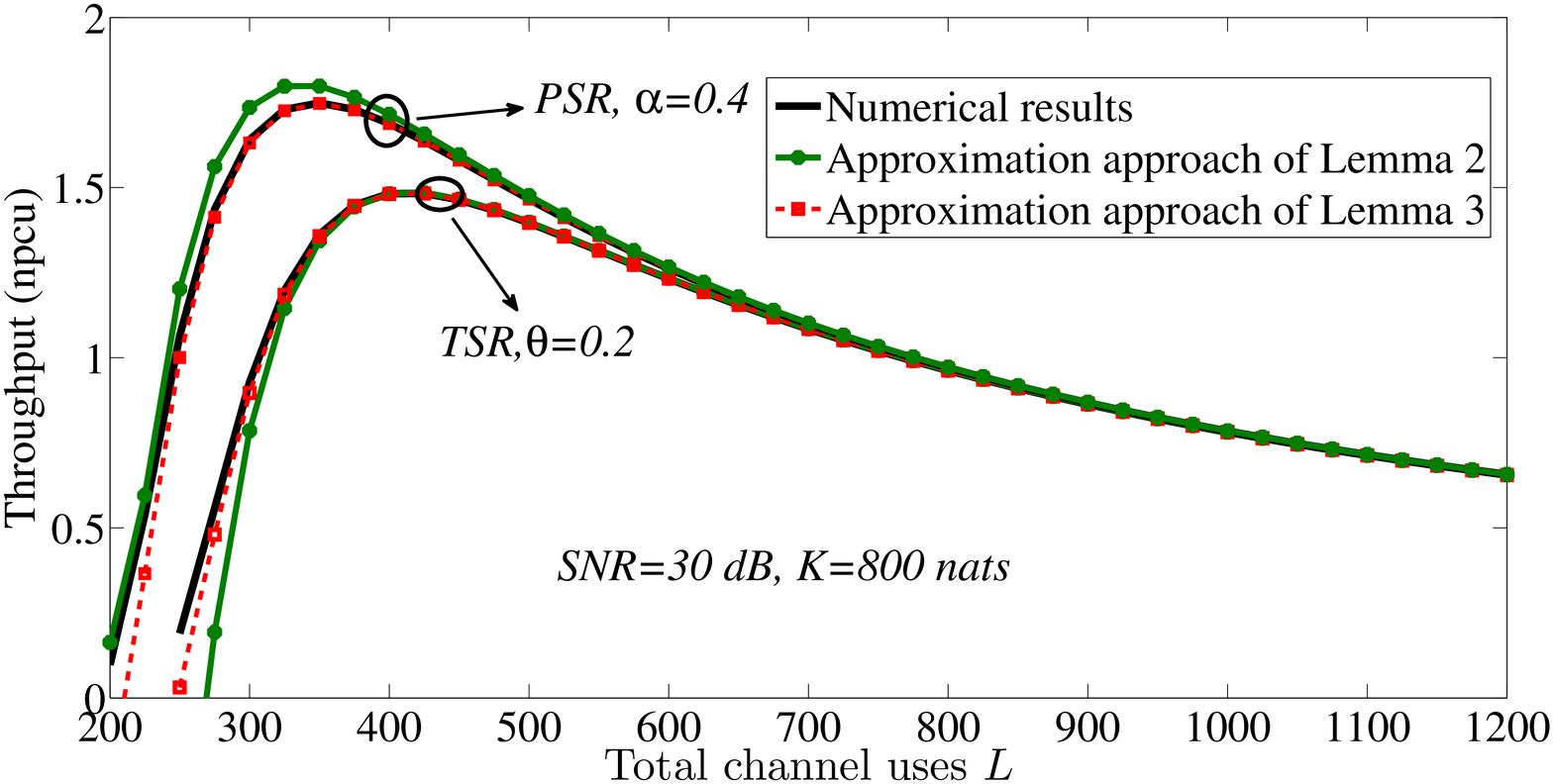}
\caption{Throughput of the PSR and TSR protocols vs. the total number of channel uses. $K=800$ nats, $\text{SNR}=30$ dB. }
\label{fig:thr_L}
\centering     %%% not \center
{\label{fig:sys_mod}\includegraphics[width=3.4in, height=1.55in]{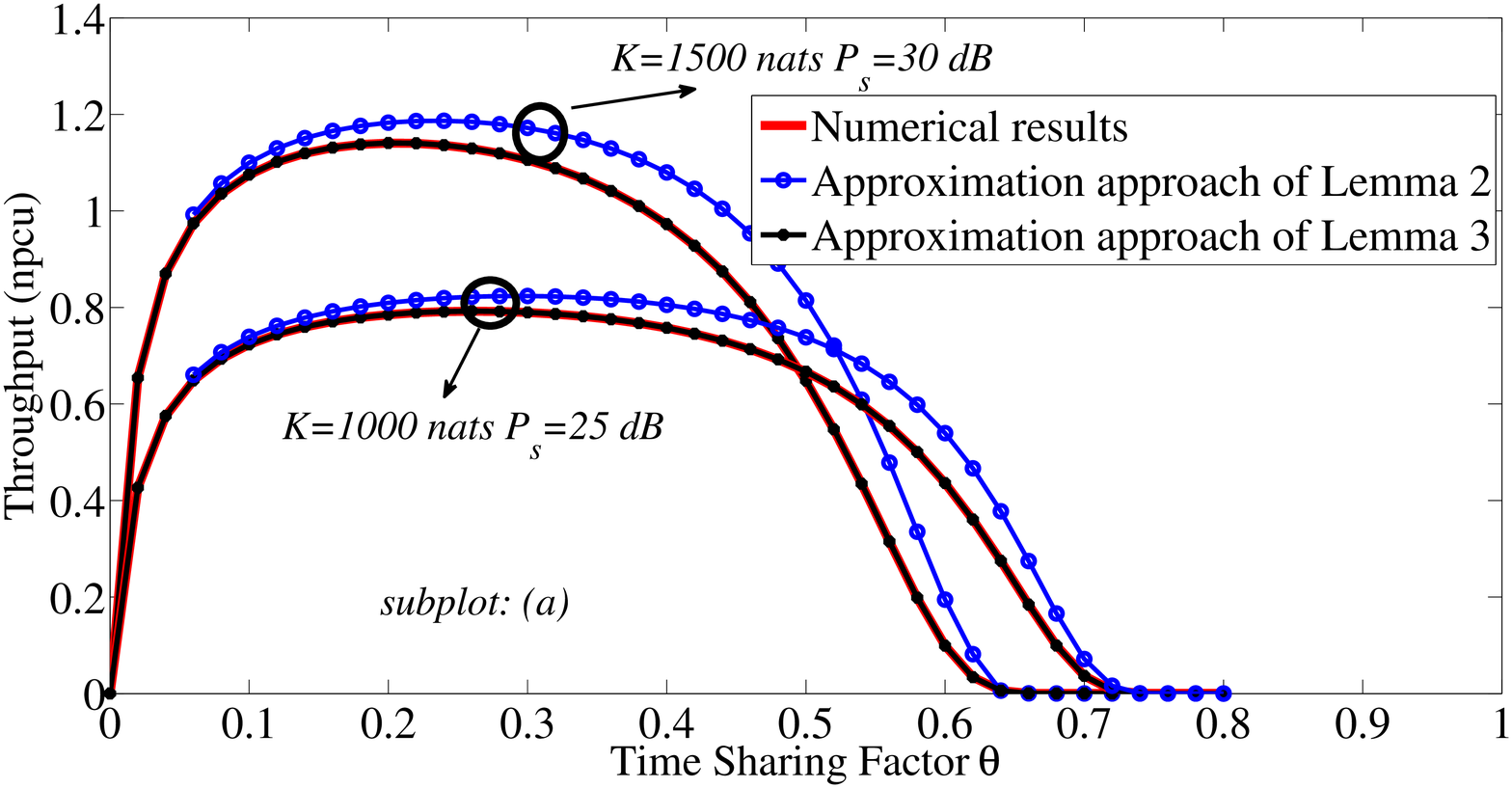}}
{\label{fig:block_patt}\includegraphics[width=3.4in, height=1.55in]{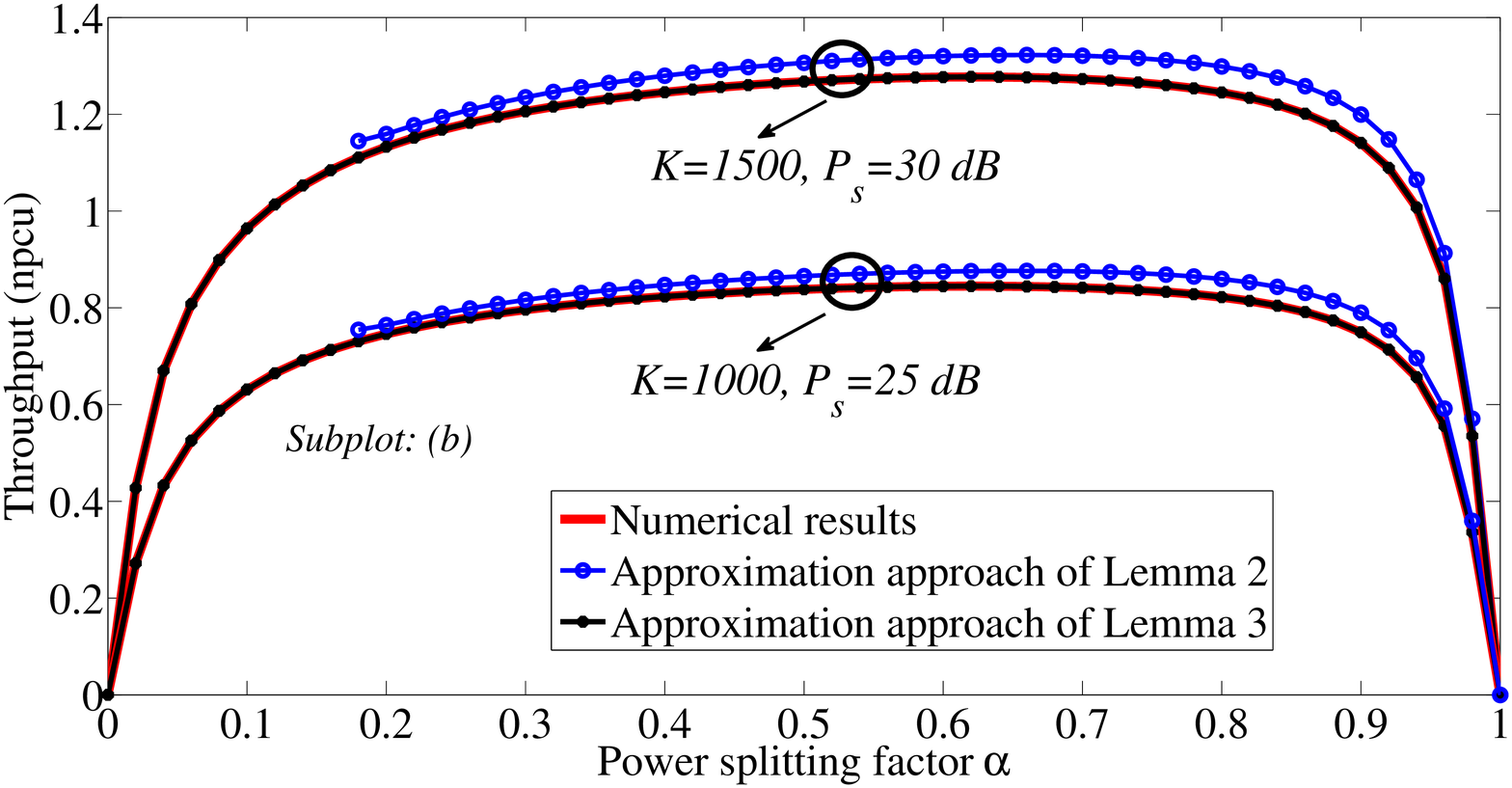}}
\caption{Throughput of the TSR (PSR) vs. time sharing (power splitting) factor. Subplot (a) shows the throughput of the TSR protocol. Subplot (b) shows the throughput of the PSR protocol. $L$=1000.  } 
\label{fig:thr_ceof}
\centering     %%% not \center
\includegraphics[width=3.4in, height=1.6in]{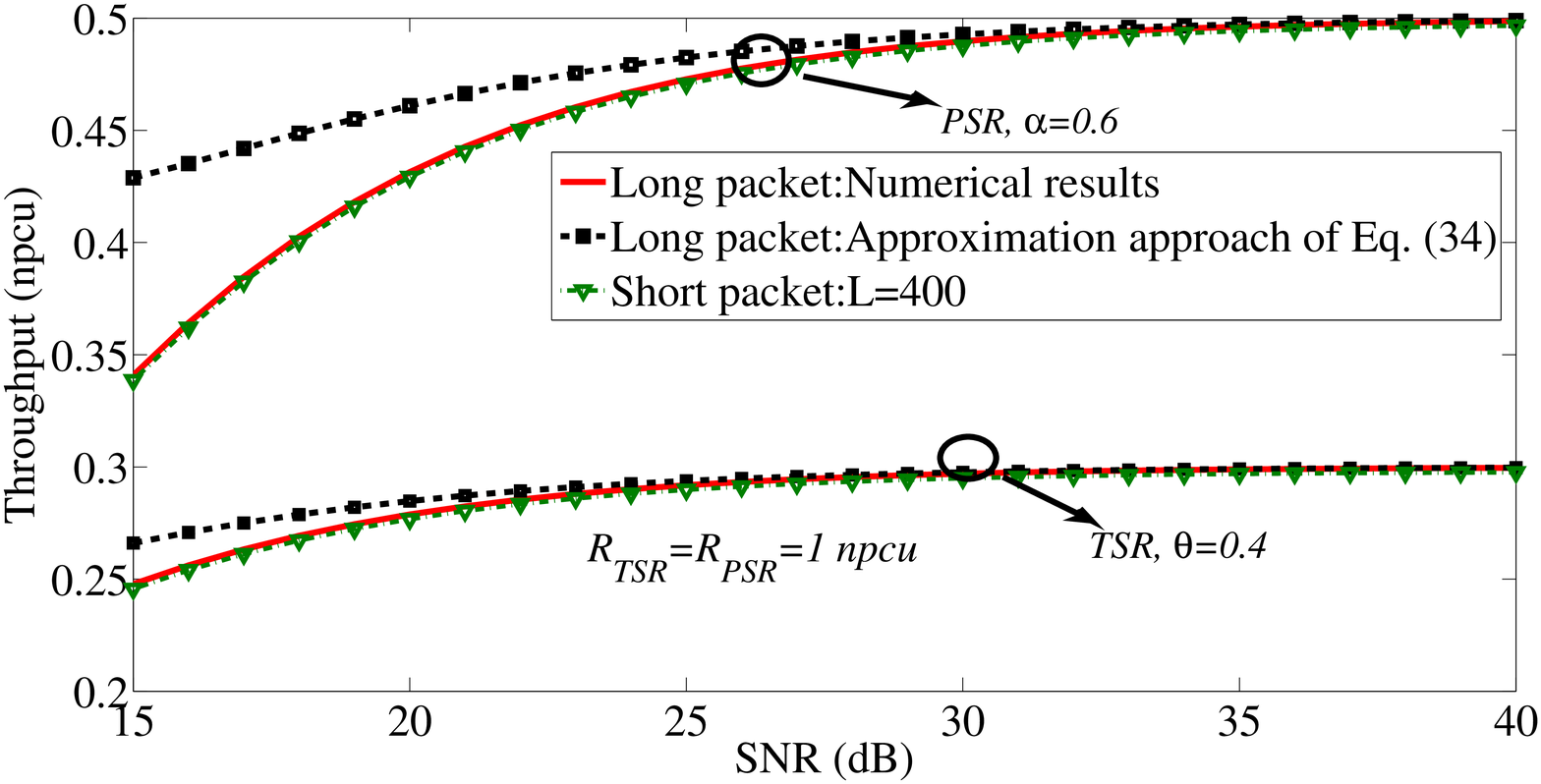}
\caption{Throughput vs. SNR. $R_{\text{TSR}}=R_{\text{PSR}}=1$ npcu. }
\vspace{-5mm}
\label{fig:thr_dif}
\end{figure}
\section{Conclusion}\label{sec:conc}
In this paper, we studied the outage probability and the throughput of an AF relay network with wireless information and energy transfer when codewords are of finite length. Two energy and information transfer protocols, namely, TSR and PSR, were considered. Accurate approximations  of the outage probability and the throughput were presented and validated with numerical results. Then, we investigated the impact of blocklength on the outage probability at high SNRs.  Our analytical and numerical results showed that, for a given number of information nats, the optimal throughput is achieved by using finite-length packets. Our results also demonstrated that the optimal selection of the time sharing and power splitting factors plays a key role in achieving high throughput and low outage probability in the PSR and TSR protocols. 
\bibliographystyle{IEEEtran} %lic.bst is the style file
\bibliography{ref_icc_paper}

% Generated by IEEEtran.bst, version: 1.14 (2015/08/26)
\begin{thebibliography}{10}
\providecommand{\url}[1]{#1}
\csname url@samestyle\endcsname
\providecommand{\newblock}{\relax}
\providecommand{\bibinfo}[2]{#2}
\providecommand{\BIBentrySTDinterwordspacing}{\spaceskip=0pt\relax}
\providecommand{\BIBentryALTinterwordstretchfactor}{4}
\providecommand{\BIBentryALTinterwordspacing}{\spaceskip=\fontdimen2\font plus
\BIBentryALTinterwordstretchfactor\fontdimen3\font minus
  \fontdimen4\font\relax}
\providecommand{\BIBforeignlanguage}[2]{{%
\expandafter\ifx\csname l@#1\endcsname\relax
\typeout{** WARNING: IEEEtran.bst: No hyphenation pattern has been}%
\typeout{** loaded for the language `#1'. Using the pattern for}%
\typeout{** the default language instead.}%
\else
\language=\csname l@#1\endcsname
\fi
#2}}
\providecommand{\BIBdecl}{\relax}
\BIBdecl

\bibitem{tutorwet}
X.~Lu, P.~Wang, D.~Niyato, D.~I. Kim, and Z.~Han, ``Wireless networks with {RF}
  energy harvesting: A contemporary survey,'' \emph{IEEE Commun. Surveys
  Tuts.}, vol.~17, no.~2, pp. 757 -- 789, May 2015.

\bibitem{cirs}
X.~Zhou, R.~Zhang, and C.~K. Ho, ``Wireless information and power transfer:
  Architecture design and rate-energy tradeoff,'' \emph{IEEE Trans. Commun.},
  vol.~61, no.~11, pp. 4754 -- 4767, Nov. 2013.

\bibitem{nasir13}
A.~Nasir, Z.~Xiangyun, S.~Durrani, and R.~Kennedy, ``Relaying protocols for
  wireless energy harvesting and information processing,'' \emph{IEEE Trans.
  Wireless Commun.}, vol.~12, no.~7, pp. 3622 -- 3636, July 2013.

\bibitem{jad2}
I.~Krikidis, S.~Timotheou, and S.~Sasaki, ``{RF} energy transfer for
  cooperative networks: Data relaying or energy harvesting?'' \emph{IEEE
  Commun. Lett}, vol.~16, no.~11, pp. 1772 -- 1775, Sep. 2012.

\bibitem{poor}
Z.~Ding, S.~Perlaza, I.~Esnaola, and H.~Poor, ``Power allocation strategies in
  energy harvesting wireless cooperative networks,'' \emph{IEEE Trans. Wireless
  Commun.}, vol.~13, no.~2, pp. 846 -- 860, Feb. 2014.

\bibitem{poor_rand}
Z.~Ding, I.~Krikidis, B.~Sharif, and H.~Poor, ``Wireless information and power
  transfer in cooperative networks with spatially random relays,'' \emph{IEEE
  Trans. Wireless Commun.}, vol.~13, no.~8, pp. 4440 -- 4453, Aug. 2014.

\bibitem{finite1}
I.~Krikidis, ``Relay selection in wireless powered cooperative networks with
  energy storage,'' \emph{IEEE J. Sel. Areas Commun.}, vol.~33, no.~12, pp.
  2596 -- 2610, Oct. 2015.

\bibitem{delay1}
A.~Osseiran \emph{et~al.}, ``Scenarios for {5G} mobile and wireless
  communications: the vision of the {METIS} project,'' \emph{IEEE Commun. Mag},
  vol.~52, no.~5, pp. 4440 -- 4453, May 2014.

\bibitem{bl_1}
Y.~Polyanskiy, H.~V. Poor, and S.~Verdu, ``Channel coding rate in the finite
  blocklength regime,'' \emph{IEEE Trans. Inf. Theory}, vol.~56, no.~5, pp.
  2307 -- 2359, May 2010.

\bibitem{mak2}
B.~Makki, T.~Svensson, and M.~Zorzi, ``Finite block-length analysis of the
  incremental redundancy {HARQ},'' \emph{IEEE Wireless Commun. Lett}, vol.~3,
  no.~5, pp. 529 -- 532, Aug. 2014.

\bibitem{behmain}
------, ``Finite block-length analysis of spectrum sharing networks using rate
  adaptation,'' \emph{IEEE Trans. Commun.}, vol.~63, no.~8, pp. 2823 -- 2835,
  Aug. 2015.

\bibitem{fb_rele}
H.~Yulin, J.~Gross, and A.~Schmeink, ``On the capacity of relaying with finite
  blocklength,'' \emph{IEEE Trans. Veh. Technol.}, vol.~65, no.~3, pp. 1790 --
  1794, March 2016.

\bibitem{rel_snr}
J.~Laneman, D.~Tse, and G.~W. Wornell, ``Cooperative diversity in wireless
  networks: Efficient protocols and outage behavior,'' \emph{IEEE Trans. Inf.
  Theory}, vol.~50, no.~12, pp. 3062 -- 3080, Dec. 2004.

\bibitem{af1}
I.~Krikidis, J.~Thompson, S.~Mclaughlin, and N.~Goertz, ``Amplify-and-forward
  with partial relay selection,'' \emph{IEEE Commun. Lett.}, vol.~12, no.~4,
  pp. 235 -- 237, April 2008.

\bibitem{tabebesel}
I.~S. Gradshteyn and I.~M. Ryzhik, \emph{Table of Integrals, Series and
  Products}.\hskip 1em plus 0.5em minus 0.4em\relax Academic Press, 2000.

\bibitem{ng1969table}
E.~W. Ng and M.~Geller, ``A table of integrals of the error functions,''
  \emph{Journal of Research of the National Bureau of Standards B}, vol.~73,
  no.~1, pp. 1--20, 1969.

\end{thebibliography}

\end{document}